\documentclass[11pt, svgnames]{article}
\usepackage[switch, columnwise]{lineno} 
\usepackage{amsmath, amsthm, amssymb, amsfonts, bbm, graphicx}
\usepackage[margin=1in]{geometry}
\usepackage[lastpage,user]{zref}
\usepackage{fancyhdr}
\usepackage[square, numbers]{natbib}
\usepackage{mdframed}
\usepackage{subcaption}

\usepackage{dsfont}
\usepackage{multicol}

\usepackage{setspace}

\usepackage{footmisc}
\usepackage{color, xcolor, fancybox, ascii, booktabs, colortbl, enumerate}

\graphicspath{{./figures/}{./figures_old/}}
\newcolumntype{a}{>{\columncolor{black!10}}c}

\providecommand{\keywords}[1]{\textbf{\textit{Keywords:}} #1}

\newcommand{\E}{\mathbb{E}}
\newcommand{\R}{\mathbb{R}}
\renewcommand{\P}{\mathbb{P}}

\newcommand{\one}{\mathds{1}}

\newcommand{\ppv}{\mathrm{PPV}}
\newcommand{\fpr}{\mathrm{FPR}}
\newcommand{\fnr}{\mathrm{FNR}}

\usepackage{perpage}
\MakePerPage{footnote}

\numberwithin{equation}{section}

\theoremstyle{plain}
    
    \newtheorem*{theorem*}{Theorem}
    \newtheorem{prop}{Proposition}
    \newtheorem{corollary}{Corollary}
\theoremstyle{definition}
    \newtheorem{definition}{Definition}
    \newtheorem*{definition*}{Definition}
    
\numberwithin{prop}{section}
\numberwithin{corollary}{section}

\title{Fair prediction with disparate impact: \\ {\Large A study of bias in recidivism prediction instruments}}
\author{Alexandra Chouldechova \thanks{Heinz College, Carnegie Mellon University}}

\date{Last revised: February 2017}

\begin{document}
  \maketitle
  
  
  \begin{abstract}
 Recidivism prediction instruments (RPI's) provide decision makers with an assessment of the likelihood that a criminal defendant will reoffend at a future point in time. While such instruments are gaining increasing popularity across the country, their use is attracting tremendous controversy. Much of the controversy concerns potential discriminatory bias in the risk assessments that are produced. This paper discusses several fairness criteria that have recently been applied to assess the fairness of recidivism prediction instruments.  We demonstrate that the criteria cannot all be simultaneously satisfied when recidivism prevalence differs across groups.  We then show how disparate impact can arise when a recidivism prediction instrument fails to satisfy the criterion of error rate balance. 
  \end{abstract}
  
  
\section{Introduction} \label{sec:intro}

Risk assessment instruments are gaining increasing popularity within the criminal justice system, with versions of such instruments being used or considered for use in pre-trial decision-making, parole decisions, and in some states even sentencing \citep{ali-model-code-draft-4,blomberg2010validation, 538rpi}.  In each of these cases, a high-risk classification---particularly a high-risk misclassification---may have a direct adverse impact on a criminal defendant's outcome.   If the use of RPI's is to become commonplace, it is especially important to ensure that the instruments are free from discriminatory biases that could result in unethical practices and inequitable outcomes for different groups.

In a recent widely popularized investigation conducted by a team at ProPublica, \citet{propublica2016} studied an RPI called COMPAS\footnote{COMPAS \citep{compasfaq} is a risk assessment instrument developed by Northpointe Inc..  Of the 22 scales that COMPAS provides, the Recidivism risk and Violent Recidivism risk scales are the most widely used.  The empirical results in this paper are based on decile scores coming from the COMPAS Recidivism risk scale.}, concluding that it is biased against black defendants.  The authors found that the likelihood of a non-recidivating black defendant being assessed as high risk is nearly twice that of white defendants.  Similarly, the likelihood of a recidivating black defendant being assessed as low risk is nearly half that of white defendants.  In technical terms, these findings indicate that the COMPAS instrument has considerably higher false positive rates and lower false negative rates for black defendants than for white defendants.  

ProPublica's analysis has met with much criticism from both the academic community and from the Northpointe corporation.   Much of the criticism has focussed on the particular choice of fairness criteria selected for the investigation.  \citet{floresfalse} argue that the correct approach for assessing RPI bias is instead to check for \emph{calibration}, a fairness criterion that they show COMPAS satisfies.    Northpointe in their response\cite{dieterich2016compas} argue for a still different approach that checks for a fairness criterion termed \emph{predictive parity}, which they demonstrate COMPAS also satisfies.  We provide precise definitions and a more in-depth discussion of these and other fairness criteria in Section \ref{sec:background}.  

In this paper we show that the differences in false positive and false negative rates cited as evidence of racial bias by \citet{propublica2016} are a direct consequence of applying an RPI that that satisfies predictive parity to a population in which recidivism prevalence\footnote{\emph{Prevalence}, also termed the \emph{base rate}, is the proportion of individuals who recidivate in a given population.} differs across groups.   Our main contribution is twofold.  (1) First, we make precise the connection between the predictive parity criterion and error rates in classification.  (2) Next, we demonstrate how using an RPI that has different false postive and false negative rates between groups can lead to disparate impact when individuals assessed as high risk receive stricter penalties.  Throughout our discussion we use the term \emph{disparate impact} to refer to settings where a penalty policy has unintended disproportionate adverse impact on a particular group.

It is important to bear in mind that fairness itself---along with the notion of disparate impact---is a social and ethical concept, not a statistical one.  A risk prediction instrument that is fair with respect to particular fairness criteria may nevertheless result in disparate impact depending on how and where it is used.  In this paper we consider hypothetical use cases in which we are able to directly connect particular fairness properties of an RPI to a measure of disparate impact.  We present both theoretical and empirical results to illustrate how disparate impact can arise.

\subsection{Outline of paper}

We begin in Section \ref{sec:fairness} by providing some background on several of the different fairness criteria that have appeared in recent literature.  We then proceed to demonstrate that an instrument that satisfies predictive parity cannot have equal false positive and negative rates across groups when the recidivism prevalence differs across those groups.  In Section \ref{sec:impact} we analyse a simple risk assessment-based sentencing policy and show how differences in false positive and false negative rates can result in disparate impact under this policy.  In Section \ref{sec:empirical} we back up our theoretical analysis by presenting some empirical results based on the data made available by the ProPublica investigators.  We conclude with a discussion of the issues that biased data presents for the arguments put forth in this paper.

\subsection{Data description and setup} \label{sec:intro_outline}

The empirical results in this paper are based on the Broward County data made publicly available by ProPublica\cite{propublica2016data}.  This data set contains COMPAS recidivism risk decile scores, 2-year recidivism outcomes, and a number of demographic and crime-related variables on individuals who were scored in 2013 and 2014.    We restrict our attention to the subset of defendants whose race is recorded as African-American ($b$) or Caucasian ($w$).\footnote{There are 6 racial groups represented in the data.  85\% of individuals are either African-American or Caucasian.}  After applying the same data pre-processing and filtering as reported in the ProPublica analysis, we are left with a data set on $n = 6150$ individuals, of whom $n_b = 3696$ are African-American and $n_c = 2454$ are Caucasian.

  
\section{Assessing fairness} \label{sec:fairness}

\subsection{Background} \label{sec:background}

We begin by with some notation.  Let $S = S(x)$ denote the risk score based on covariates \mbox{$X = x \in \R^p$}, with higher values of $S$ corresponding to higher levels of assessed risk.  We will interchangeably refer to $S$ as a \emph{score} or an \emph{instrument}.  For simplicity, our discussion of fairness criteria will focus on a setting where there exist just two groups.  We let $R \in \{b, w\}$ denote the group to which an individual belongs, and do not preclude $R$ from being one of the elements of $X$.  We denote the outcome indicator by $Y \in \{0 , 1\}$, with $Y = 1$ indicating that the given individual goes on to recidivate.  Lastly, we introduce the quantity $s_\mathrm{HR}$, which denotes the high-risk score threshold.  Defendants whose score $S$ exceeds $s_\mathrm{HR}$ will be referred to as \emph{high-risk}, while the remaining defendants will be referred to as \emph{low-risk}.  

  With this notation in hand, we now proceed to define and discuss several fairness criteria that commonly appear in the literature, beginning with those mentioned in the introduction.  We indicate cases where a given criterion is known to us to also commonly appear under some other name.  All of the criteria presented below can also be assessed \emph{conditionally} by further conditioning on some covariates in $X$.  We discuss this point in greater detail in Section~\ref{sec:covariates}.
  
 \begin{definition}[Calibration]
   A score $S = S(x)$ is said to be \emph{well-calibrated} if it reflects the same likelihood of recidivism irrespective of the individuals' group membership. That is, if for all values of $s$,
   \begin{equation}
   \P(Y = 1 \mid S = s, R = b) = \P(Y = 1 \mid S = s, R = w).
   \label{eq:def_calibration}
   \end{equation}
 \end{definition}
 
 Within the educational and psychological testing and assessment literature, the notion of \emph{calibration} features among the widely accepted and adopted standards for empirical fairness assessment.  In this literature, an instrument that is \emph{well-calibrated} is referred to as being \emph{free from predictive bias}.  This criterion has recently been applied to the PCRA\footnote{The Post Conviction Risk Assessment (PCRA) tool was developed by the Administrative Office of the United States Courts for the purpose of improving ``the effectiveness and efficiency of post-conviction supervision''\citep{pcra}} instrument, with initial findings suggesting that calibration is satisfied with respect race\citep{singh2013predictive, skeem2015risk}, but not with respect to gender\citep{skeem2016gender}.   In their response to the ProPublica investigation, \citet{floresfalse} verify that COMPAS is well-calibrated using logistic regression modeling. 
  
 \begin{definition}[Predictive parity]
   A score $S = S(x)$ satisfies \emph{predictive parity} at a threshold $s_\mathrm{HR}$ if the likelihood of recidivism among high-risk offenders is the same regardless of group membership. That is, if,
   \begin{equation}
   \P(Y = 1 \mid S > s_\mathrm{HR} , R = b) = \P(Y = 1 \mid S > s_\mathrm{HR}, R = w).
   \label{eq:def_predictive_parity}
   \end{equation}
 \end{definition}
 
 Predictive parity at a given threshold $s_\mathrm{HR}$ amounts to requiring that the \emph{positive predictive value} (PPV) of the classifier $\hat Y = \one_{S > s_\mathrm{HR}}$ be the same across groups.  While predictive parity and calibration look like very similar criteria, well-calibrated scores can fail to satisfy predictive parity at a given threshold.  This is because the relationship between \eqref{eq:def_predictive_parity} and \eqref{eq:def_calibration} depends on the conditional distribution of $S \mid R = r$, which can differ across groups in ways that result in PPV imbalance.  In the simple case where $S$ itself is binary, a score that is well-calibrated will also satisfy predictive parity. Northpointe's refutation\cite{dieterich2016compas} of the ProPublica analysis shows that COMPAS satisfies predictive parity for threshold choices of interest. 

 \begin{definition}[Error rate balance]
   A score $S = S(x)$ satisfies \emph{error rate balance} at a threshold $s_\mathrm{HR}$ if the false positive and false negative error rates are equal across groups. That is, if,
   \begin{align}
   &\P(S > s_\mathrm{HR} \mid  Y = 0,  R = b) = \P(S > s_\mathrm{HR} \mid  Y = 0,  R = w)  \;, \;\;  \text{ and} \label{eq:def_fpr}\\ 
   & \P(S \le s_\mathrm{HR} \mid  Y = 1,  R = b) = \P(S \le s_\mathrm{HR} \mid  Y = 1,  R = w), \label{eq:def_fnr}
   \end{align}
   where the expressions in the first line are the group-specific false positive rates, and those in the second line are the group-specific false negative rates.  
 \end{definition} 
 
 ProPublica's analysis considered a threshold of $s_\mathrm{HR} = 4$, which they showed leads to considerable imbalance in both false positive and false negative rates.  While this choice of cutoff met with some criticism, we will see later in this section that error rate imbalance persists---indeed, must persist---for any choice of cutoff at which the score satisfies the predictive parity criterion.  Error rate balance is also closely connected to the notions of \emph{equalized odds} and \emph{equal opportunity} as introduced in the recent work of \citet{hardt2016equality}.
 
 \begin{definition}[Statistical parity]
   A score $S = S(x)$ satisfies \emph{statistical parity} at a threshold $s_\mathrm{HR}$ if the proportion of individuals classified as high-risk is the same for each group. That is, if,
   \begin{equation}
   \P(S > s_\mathrm{HR} \mid  R = b) = \P(S > s_\mathrm{HR} \mid  R = w)
   \label{eq:def_stat_parity}
   \end{equation} 
 \end{definition}
 
 Statistical parity also goes by the name of \emph{equal acceptance rates}\cite{zliobaite2015relation} or \emph{group fairness}\cite{dwork2012fairness}, though it should be noted that these terms are in many cases not used synonymously.  While our discussion focusses primarily on first three fairness criteria, statistical parity is widely used within the machine learning community and may be the criterion with which many readers are most familiar\cite{calders2010three, fish2016confidence}.  Statistical parity is well-suited to contexts such as employment or admissions, where it may be desirable or required by law or regulation to employ or admit individuals in equal proportion across racial, gender, or geographical groups.  It is, however, a difficult criterion to motivate in the recidivism prediction setting, and thus will not be further considered in this work. 
 
 \subsection{Further related work}

Though the study of discrimination in decision making and predictive modeling is rapidly evolving, it also has a long and rich multidisciplinary history.   \citet*{romei2014multidisciplinary} provide an excellent overview of some of the work in this broad subject area.  The recent work of \citet*{barocas2016big} offers a broad examination of algorithmic fairness framed within the context of anti-discrimination laws governing employment practices.  \citet{hannah2013actuarial}, \citet{skeem2013risk}, and \citet{monahan2016risk} examine legal and ethical issues relating specifically to the use of risk assessment instruments in sentencing, citing the potential for race and gender discrimination as a major concern.

In work concurrent with our own, several other researchers have also investigated the compatibility of different notions of fairness.  \citet{kleinberg2016inherent} show that calibration cannot be satisfied simultaneously with the fairness criteria of \emph{balance for the negative class} and \emph{balance for the positive class}.  Translated into the present context, the latter criteria require that the average score assigned to non-recidivists (the negative class) should be the same for both groups, and that the same should hold among recidivists (the positive class).
The work of \citet{corbett-davies-2016} closely parallels the results that we present in Section \ref{sec:incompatibility}, reaching the same conclusion regarding the incompatibility of predictive parity and error rate balance in the setting of unequal prevalence.

  \subsection{Predictive parity, false positive rates, and false negative rates} \label{sec:incompatibility}

  \begin{figure}[!hb]
   \centering
   \includegraphics[width = 0.5\linewidth]{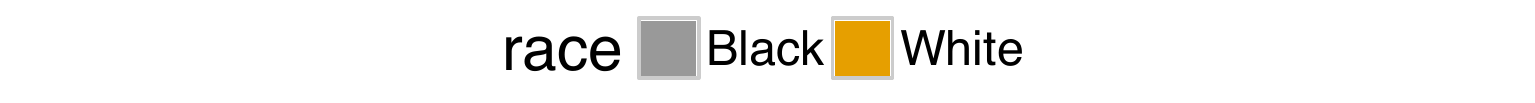} \\
   \begin{subfigure}[t]{0.48\linewidth}
     \includegraphics[width=\textwidth]{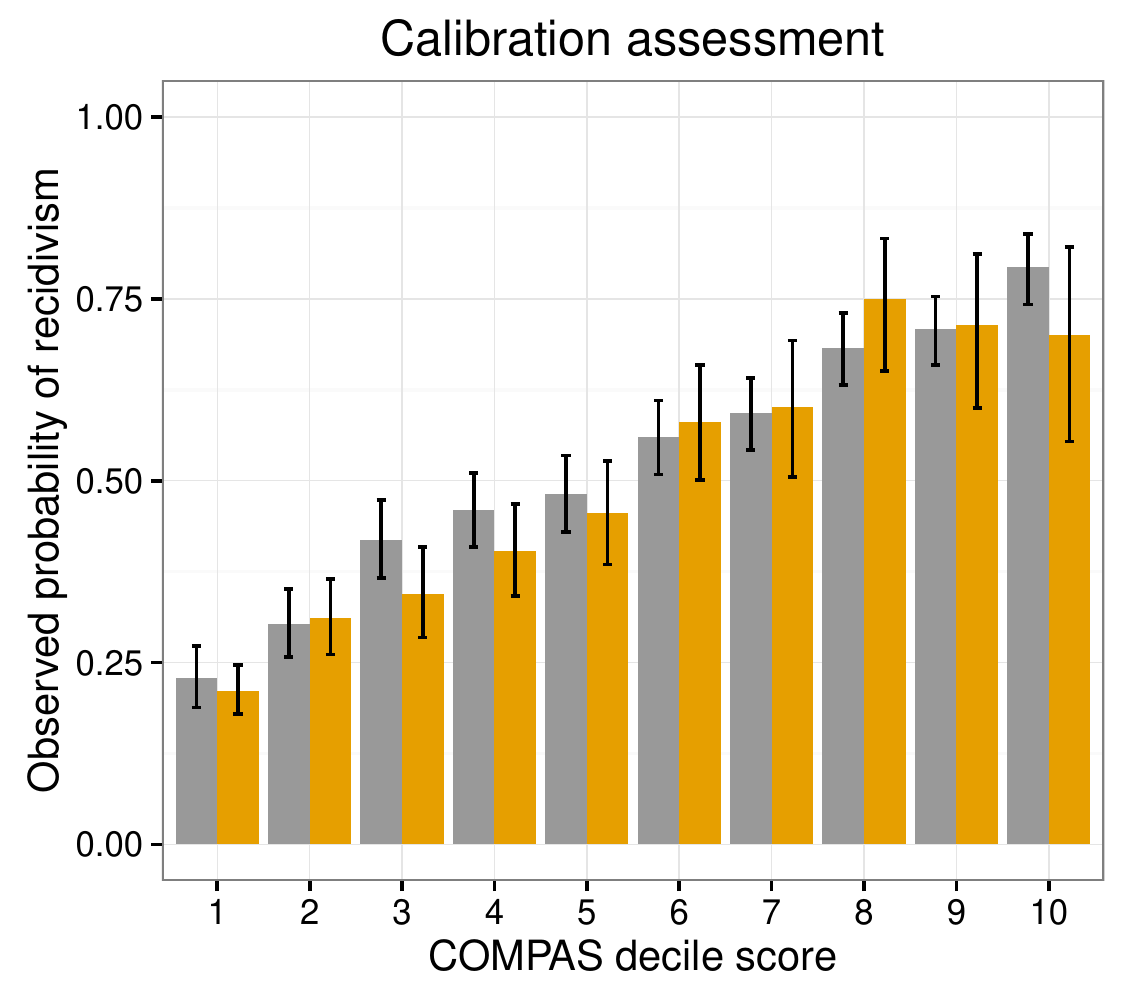}
     \caption{Bars represent empirical estimates of the expressions in \eqref{eq:def_calibration}: $\P(Y = 1 \mid S = s, R = r)$ for decile scores $s \in \{1, \ldots, 10\}.$}
   \end{subfigure}
   \hspace{1em}
   \begin{subfigure}[t]{0.48\linewidth}
     \includegraphics[width=\textwidth]{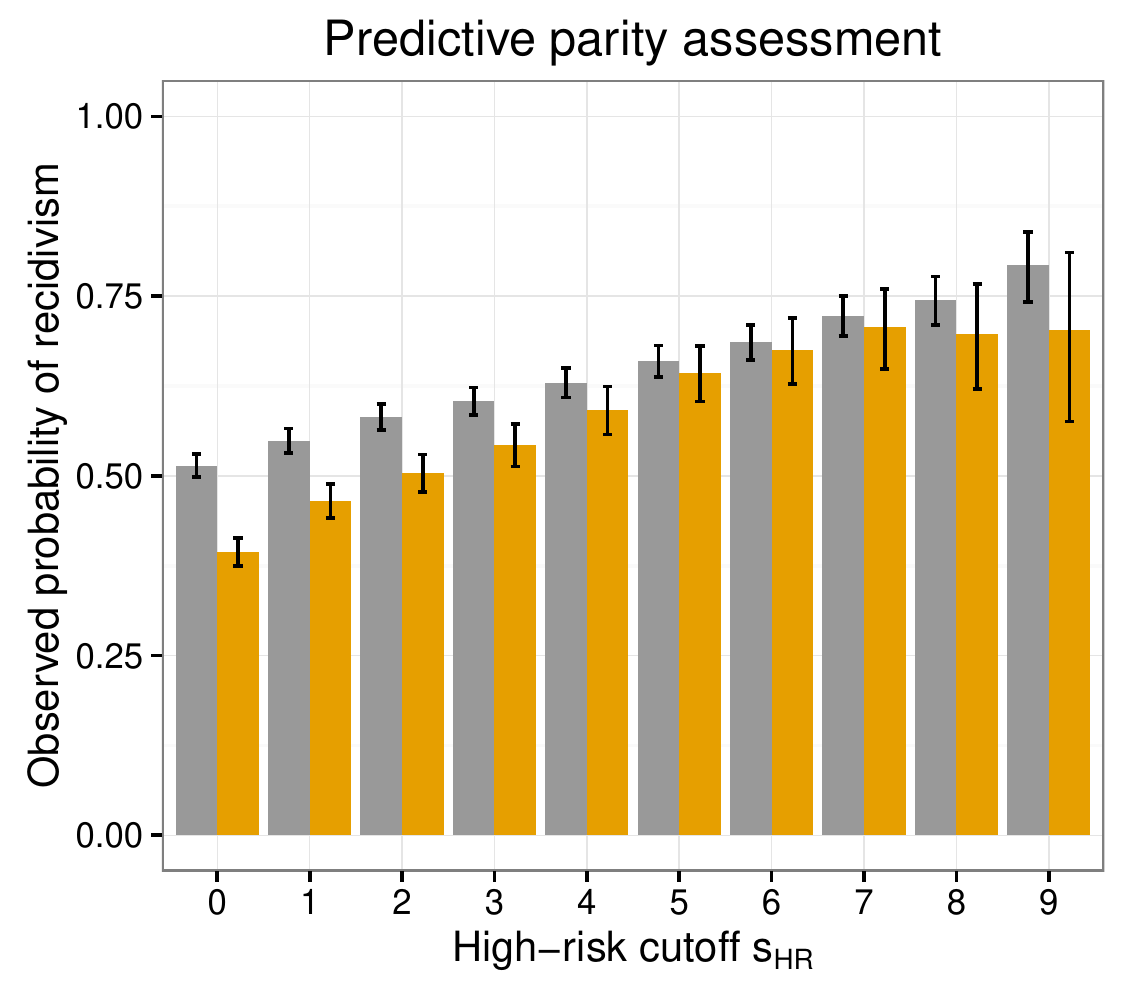}
     \caption{Bars represent empirical estimates of the expressions in \eqref{eq:def_predictive_parity}: $\P(Y = 1 \mid S > s_\mathrm{HR} , R = r)$ for values of the high-risk cutoff $s_{\mathrm{HR}} \in \{0, \ldots, 9\}$}
   \end{subfigure} \\
   \begin{subfigure}[t]{0.48\linewidth}
     \includegraphics[width=\textwidth]{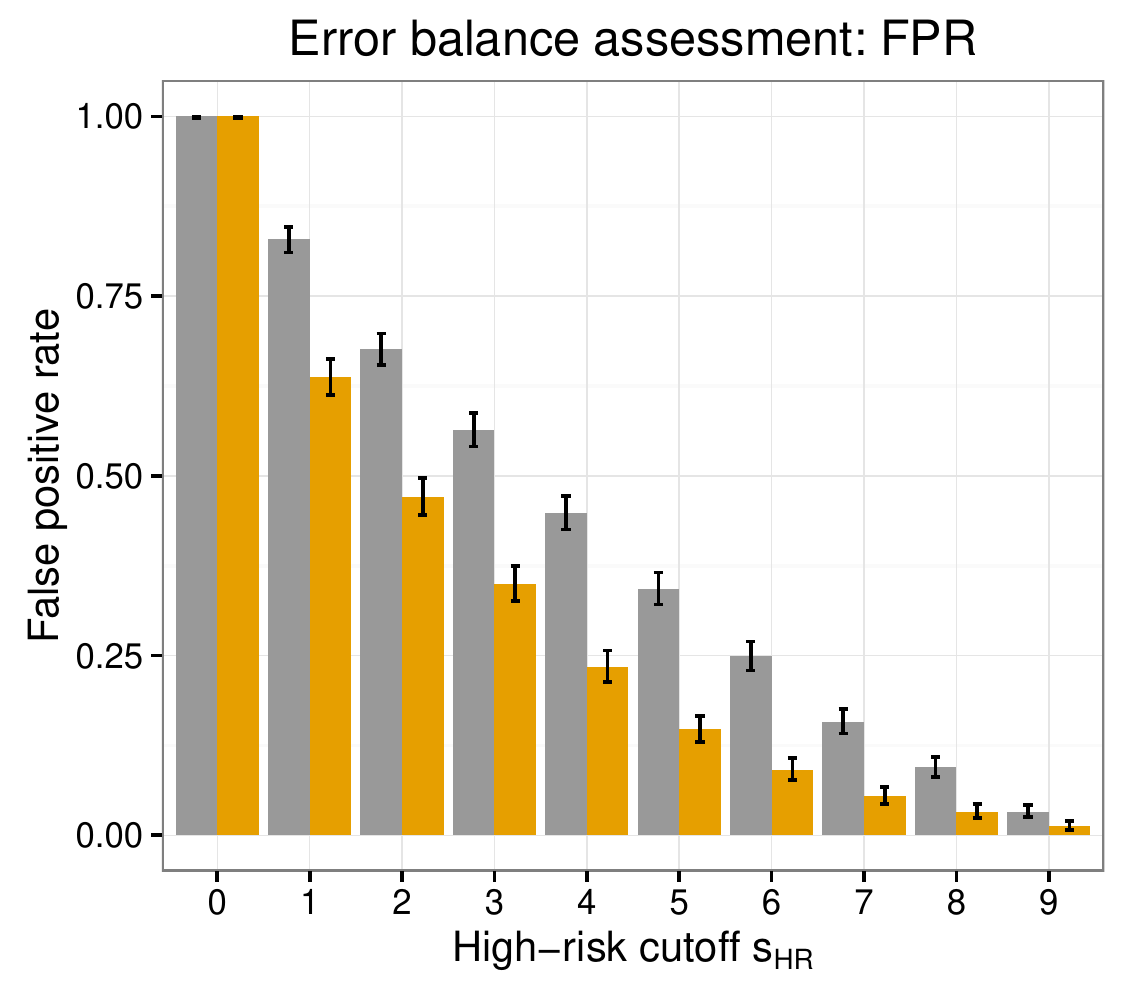}
     \caption{Bars represent observed false positive rates, which are empirical estimates of the expressions in \eqref{eq:def_fpr}: $\P(S > s_\mathrm{HR} \mid Y = 0 , R = r)$ for values of the high-risk cutoff $s_{\mathrm{HR}} \in \{0, \ldots, 9\}$}
     \label{subfig:fpr}
   \end{subfigure}
   \hspace{1em}
   \begin{subfigure}[t]{0.48\linewidth}
     \includegraphics[width=\textwidth]{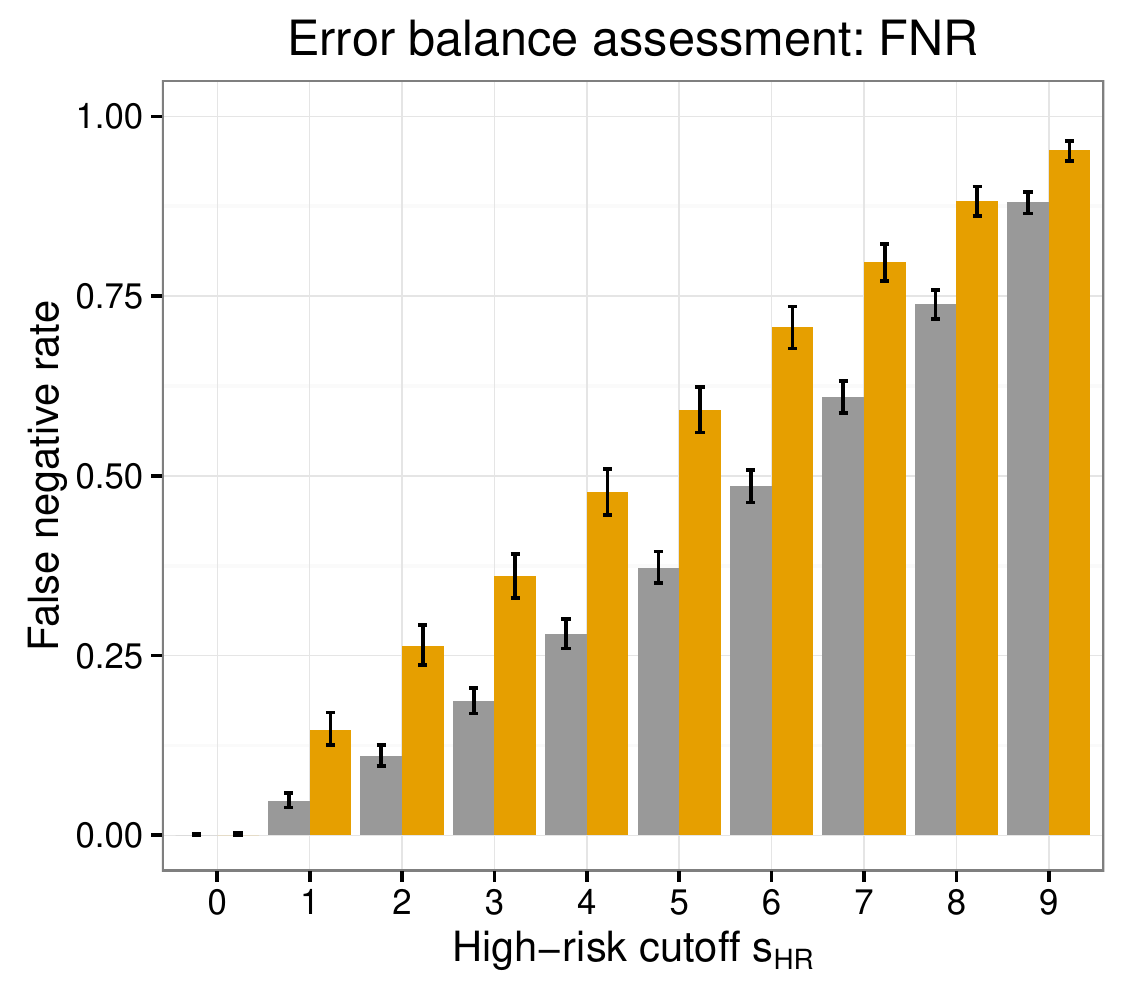}
     \caption{Bars represent observed false negative rates, which are empirical estimates of the expressions in \eqref{eq:def_fnr}: $\P(S \le s_\mathrm{HR} \mid Y = 1, R = r)$ for values of the high-risk cutoff $s_{\mathrm{HR}} \in \{0, \ldots, 9\}$}
     \label{subfig:fnr}
   \end{subfigure}
 
   \vspace{1em}
   \caption{Empirical assessment of the COMPAS RPI according to three of the fairness criteria presented in Section \ref{sec:background}.  Error bars represent 95\% confidence intervals.  These Figures confirm that COMPAS is (approximately) well-calibrated, satisfies predictive parity for high-risk cutoff values of 4 or higher, but fails to have error rate balance.}
   \label{fig:fairness_assessment}
  \end{figure}
 
In this section we present our first main result, which establishes that predictive parity is incompatible with error rate balance when prevalence differs across groups.  To better motivate the discussion, we begin by presenting an empirical fairness assessment of the COMPAS RPI.  Figure~\ref{fig:fairness_assessment} shows plots of the observed recidivism rates and error rates corresponding to the fairness notions of calibration, predictive parity, and error rate balance.   We see that the COMPAS RPI is (approximately) well-calibrated, and also satisfies predictive parity provided that the high-risk cutoff $s_\mathrm{HR}$ is $4$ or greater. However, COMPAS fails on both false positive and false negative error rate balance across the range of high-risk cutoffs. 
 
\citet{propublica2016} focussed on a high-risk cutoff of $s_\mathrm{HR} = 4$ for their analysis, which some critics have argued is too low, suggesting that $s_\mathrm{HR} = 7$ is more suitable.  As can be seen from Figures  \ref{subfig:fpr} and \ref{subfig:fnr}, significant error rate imbalance persists at this cut-off as well.   Moreover, the error rates achieved at so high a cutoff are at odds with evidence suggesting that the use of RPI's is of interest in settings where false negatives have a higher cost than false positives, with relative cost estimates ranging from 2.6  to upwards of 15. \citep{nij2013, psc2013}
 
 As we now proceed to show, the error rate imbalance exhibited by COMPAS is not a coincidence, nor can it be remedied in the present context.  When the recidivism prevalence--i.e., the base rate $\P(Y = 1 \mid R = r)$---differs across groups, any instrument that satisfies predictive parity at a given threshold $s_\mathrm{HR}$ \emph{must} have imbalanced false positive or false negative errors rates at that threshold.  To understand why predictive parity and error rate balance are mutually exclusive in the setting of unequal recidivism prevalence, it is instructive to think of how these quantities are all related.  
 
 Given a particular choice of $s_\mathrm{HR}$, we can summarize an instrument's performance in terms of a confusion matrix, as shown in Table~\ref{tab:confusion} below.
 \begin{table}[h]
   \centering 
      \begin{tabular}{|r|rr|}
        \hline
       & Low-Risk & High-Risk \\ 
        \hline
      $Y = 0$ & TN & FP \\ 
        $Y = 1$ & FN & TP \\ 
         \hline
      \end{tabular} \\
      \vspace{0.5em}
      \caption{T/F denote True/False and N/P denote Negative/Positive. For instance, FP is the number of false positives: individuals who are classified as high-risk but who do not reoffend.}
      \label{tab:confusion}
 \end{table}

\vspace{-1em}

\noindent All of the fairness metrics presented in Section \ref{sec:background} can be thought of as imposing constraints on the values (or the distribution of values) in this table.  Another constraint---one that we have no direct control over---is imposed by the recidivism prevalence within groups.  It is not difficult to show that the prevalence ($p$), positive predictive value ($\ppv$), and false positive and negative error rates ($\fpr$, $\fnr$) are related via the equation
 \begin{equation}
   \fpr = \frac{p}{1 - p} \frac{1 - \ppv}{\ppv} ( 1 - \fnr).  
   \label{eq:fpr_fnr}
 \end{equation}
From this simple expression we can see that if an instrument satisfies predictive parity---that is, if the PPV is the same across groups---but the prevalence differs between groups, the instrument cannot achieve equal false positive and false negative rates across those groups.   
  
 This observation enables us to better understand why we observe such large discrepancies in FPR and FNR between black and white defendants in Figure~\ref{fig:fairness_assessment}. The recidivism rate among black defendants in the data is 51\%, compared to 39\% for White defendants.  Thus at any threshold $s_\mathrm{HR}$ where the COMPAS RPI satisfies predictive parity, equation \eqref{eq:fpr_fnr} tells us that some level of imbalance in the error rates must exist.  Since not all of the fairness criteria can be satisfied at the same time, it becomes important to understand the potential impact of failing to satisfy particular criteria.  This question is explored in the context of a hypothetical risk-based sentencing framework in the next section.

\section{Assessing impact} \label{sec:impact}

In this section we show how differences in false positive and false negative rates can result in disparate impact under policies where a high-risk assessment results in a stricter penalty for the defendant.  Such situations may arise when risk assessments are used to inform bail, parole, or sentencing decisions.  In Pennsylvania and Virginia, for instance, statutes permit the use of RPI's in sentencing, provided that the sentence ultimately falls within accepted guidelines\cite{ali-model-code-draft-4}.  We use the term ``penalty'' somewhat loosely in this discussion to refer to outcomes both in the pre-trial and post-conviction phase of legal proceedings. For instance, even though pre-trial outcomes such as the amount at which bail is set are not punitive in a legal sense, we nevertheless refer to bail amount as a ``penalty'' for the purpose of our discussion.  

There are notable cases where RPI's are used for the express purpose of informing risk reduction efforts.  In such settings, individuals assessed as high risk receive what may be viewed as a benefit rather than a penalty.  The PCRA score, for instance, is intended to support precisely this type of decision-making at the federal courts level \citep{skeem2015risk}.  Our analysis in this section specifically addresses use cases where high-risk individuals receive stricter penalties.

To begin, consider a setting in which guidelines indicate that a defendant is to receive a penalty $t_\mathrm{min} \le T \le t_\mathrm{max}$.  A very simple risk-based approach, which we will refer to as the MinMax\footnote{The term MinMax as used throughout this paper has no intended connection the decision-theoretic notion of minimax decision rules.  Min and Max in this context refer to the minimum and maximum allowable sentences as stipulated by sentencing guidelines.} policy, would be to assign penalties as follows:
\begin{equation}
 T_{\mathrm{MinMax}}(s) =  \begin{cases}
      t_{\mathrm{min}} & \text{if } s > s_\mathrm{HR} \\
      t_{\mathrm{max}} & \text{if } s < s_\mathrm{HR}
      \end{cases}.
\end{equation}

In this simple setting, we can precisely characterize the extent of disparate impact in terms of recognizable quantities. Our analysis will focus on the quantity 
\[
\Delta = \Delta(y_1, y_2) \equiv \E(T \mid R = b, Y = y_1) - \E(T \mid R = w, Y = y_2),
\]
which is the expected difference in sentence duration between defendants in different groups, with potentially different outcomes $y_1, y_2 \in \{0, 1\}$.  $\Delta$ is taken to serve as our the measure of disparate impact.

\begin{prop}
  The expected difference in penalty under the MinMax policy is given by
  \begin{align*}
  \Delta &\equiv \E(T \mid R = b, Y = y_1) - \E(T \mid R = w, Y = y_2) \\ 
  &= (t_\mathrm{max}- t_\mathrm{min})\big(\P(S > s_\mathrm{HR} \mid R = b, Y = y_1) - \P(S > s_\mathrm{HR} \mid R = w, Y = y_2)\big)
  \end{align*}
  \label{prop:main}
\end{prop}

\vspace{-2em}
\noindent  A proof can be found in Appendix~\ref{appendix:proofs}. We will discuss two immediate Corollaries of this result.

\begin{corollary}[Non-Recidivists]
  Among individuals who \emph{do not recidivate}, the difference in average penalty under the MinMax policy is 
  \begin{equation}
    \Delta = (t_\mathrm{max} - t_\mathrm{min})(\fpr_b - \fpr_w),
    \label{eq:survive}
  \end{equation}
  where $\fpr_r$ denotes the false positive rate among individuals in group $R = r$.  
  \label{cor:nonrecid}
\end{corollary}
\begin{corollary}[Recidivists]
  Among individuals who \emph{recidivate}, the difference in average penalty under the MinMax policy is 
  \begin{equation}
    \Delta = (t_\mathrm{max} - t_\mathrm{min})(\fnr_w - \fnr_b),
    \label{eq:recid}
  \end{equation}
  where $\fnr_r$ denotes the false negative rate among individuals in group $R = r$.  
  \label{cor:recid}
\end{corollary}

When using an RPI that satisfies predictive parity in populations where recidivism prevalence differs across groups, it will generally be the case that the higher recidivism prevalence group will have a higher FPR and lower FNR.  From equations~\eqref{eq:survive} and \eqref{eq:recid}, we can see that this would on average result in greater penalties for defendants in the higher prevalence group, both among recidivists and non-recidivists.   

An interesting special case to consider is one where $t_\mathrm{min} = 0$.  This could arise in sentencing decisions for offenders convicted of low-severity crimes who have good prior records.  In such cases, so-called restorative sanctions may be imposed as an alternative to a period of incarceration. If we further take $t_\mathrm{max} = 1$, then $\E T = \P(T \neq 0)$, which can be interpreted as the probability that a defendant receives a sentence imposing some period of incarceration.  

It is easy to see that in such settings a non-recidivist in group $b$ is $\fpr_b /  \fpr_w$ times more likely to be incarcerated compared to a non-recidivist in group $w$.\footnote{We are overloading notation in this expression:  Here, \mbox{$\fpr_r = \P(\mathrm{HR} \mid R = r, t_L = 0)$}, similarly for $\fnr_r$.} This naturally raises the question of whether overall differences in error rates are observed to persist across more granular subpopulations, such as the subset of individuals eligible for restorative sanctions.  We explore this question in the section below.

\subsection{Conditioning on other covariates} \label{sec:covariates}

One might expect that differences in false positive rates are largely attributable to the subset of defendants who are charged with more serious offenses and who have a larger number of prior arrests/convictions.  While it is true that the false positive rates within both racial groups are higher for defendants with worse criminal histories, considerable between-group differences in these error rates persist across low prior count subgroups.  Figure~\ref{fig:fpr_prior} shows plots of false positive rates across different ranges of prior count for all defendants and also for the subset charged with a misdemeanor offense, which is the lowest severity criminal offense category.  As one can see, differences in false positive rates between Black defendants and White defendants persist across prior record subgroups.

In general, all of the theoretical results presented in this section extend to the setting where we further condition on the covariates $X$.  The main difference is that all classification metrics would need to be evaluated conditional on $X$.  For instance, assuming that $t_\mathrm{min}$ and $t_\mathrm{max}$ are constant on a set $\mathcal{X},$ Corollary~\ref{cor:nonrecid} would say that the difference in average penalty under the MinMax policy among non-recidivists for whom $X \in \mathcal{X}$ is given by
\begin{align}
\Delta &=  (t_\mathrm{max} - t_\mathrm{min})\left( \fpr_b(\mathcal{X}) - \fpr_w(\mathcal{X}) \right)  \\
&\equiv(t_\mathrm{max} - t_\mathrm{min}) \left( \P(S > s_\mathrm{HR} \mid R = b, Y = 0, X \in \mathcal{X}) - \P(S > s_\mathrm{HR} \mid R = w, Y = 0, X \in \mathcal{X}) \right).
  \label{eq:conditioning}
\end{align}
\indent  The false positive rates shown in Figure~\ref{fig:fpr_prior}(a) correspond precisely to the quantities $\fpr_r(\mathcal{X})$ for choices of $\mathcal{X}$ given by different prior record count bins.  The leftmost bars correspond to taking $\mathcal{X} = \{\# \text{priors} = 0\}$.  Similarly the leftmost bars in Figure~\ref{fig:fpr_prior}(a) correspond to taking $\mathcal{X} = \{\# \text{priors} = 0, \text{charge degree} = M\}$.  In Appendix \ref{appendix:logistic} we present a logistic regression analysis showing that significant differences in false positive rates persist even after adjusting for a number of other recidivism-related covariates.

\begin{figure}[!ht]
 \centering
 \begin{subfigure}[t]{0.48\linewidth}
   \includegraphics[width=\textwidth]{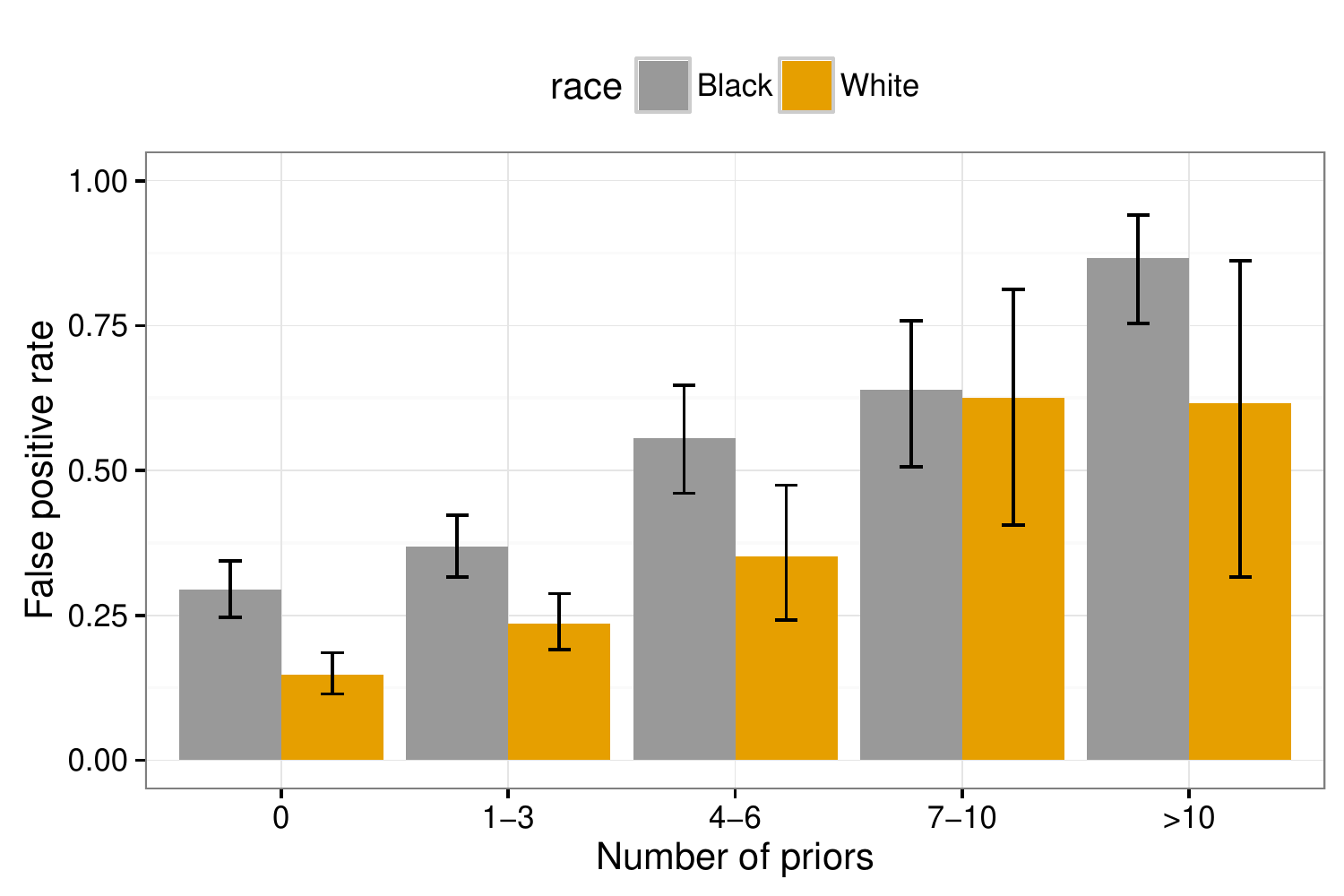}
   \caption{All defendants.}
 \end{subfigure}
 \hspace{1em}
 \begin{subfigure}[t]{0.48\linewidth}
   \includegraphics[width=\textwidth]{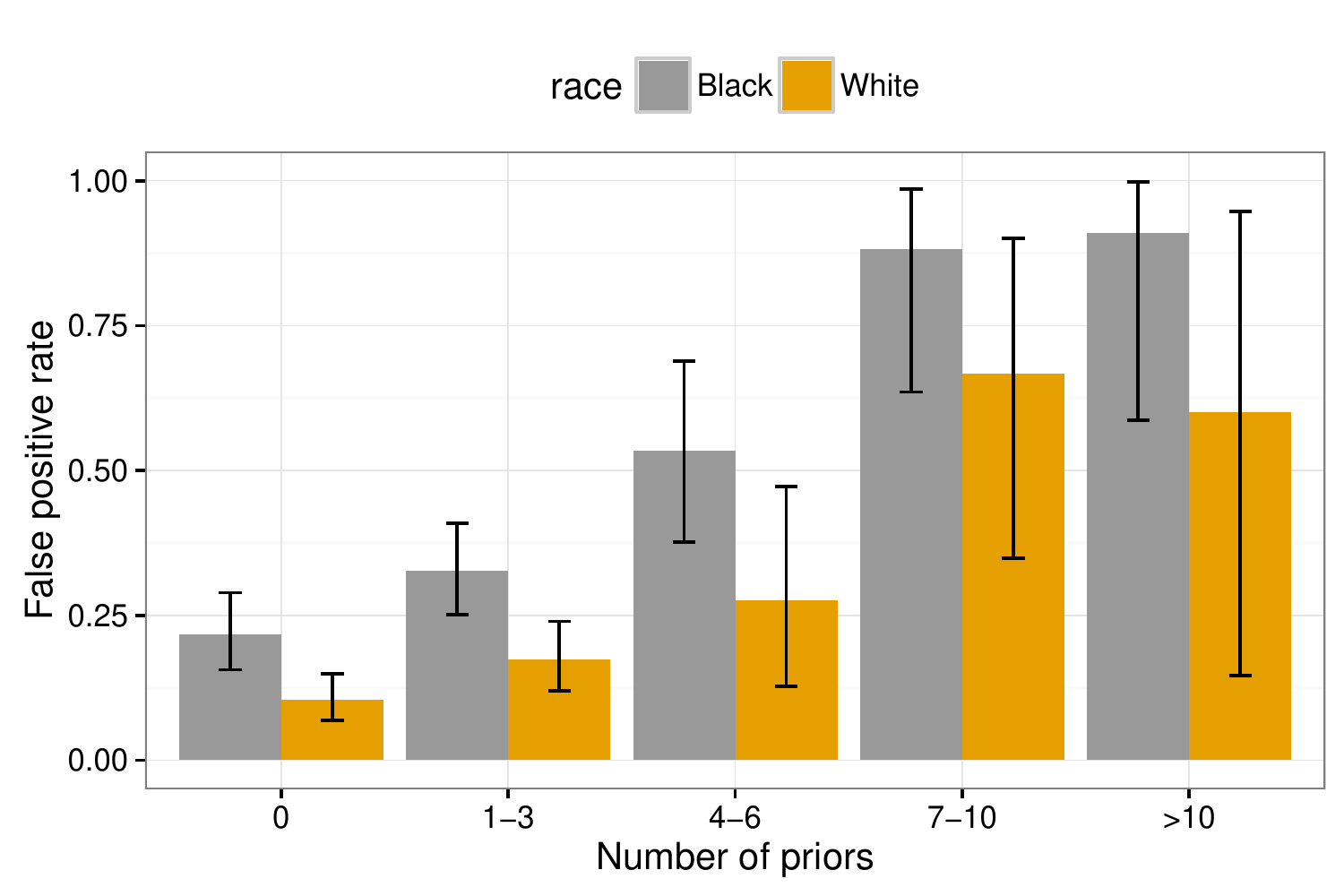}
   \caption{Defendants charged with a Misdemeanor offense. }
 \end{subfigure}
 
 \vspace{1em}
 \caption{False positive rates across prior record count.  Plot is based on assessing a defendant as ``high-risk'' if their COMPAS decile score is $> s_\mathrm{HR} = 4$. Error bars represent 95\% confidence intervals.}
 \label{fig:fpr_prior}
\end{figure}

\subsection{Connections to measures of differences in distribution}

 \begin{figure}[t]
    \centering
    \includegraphics[width = 0.75\linewidth]{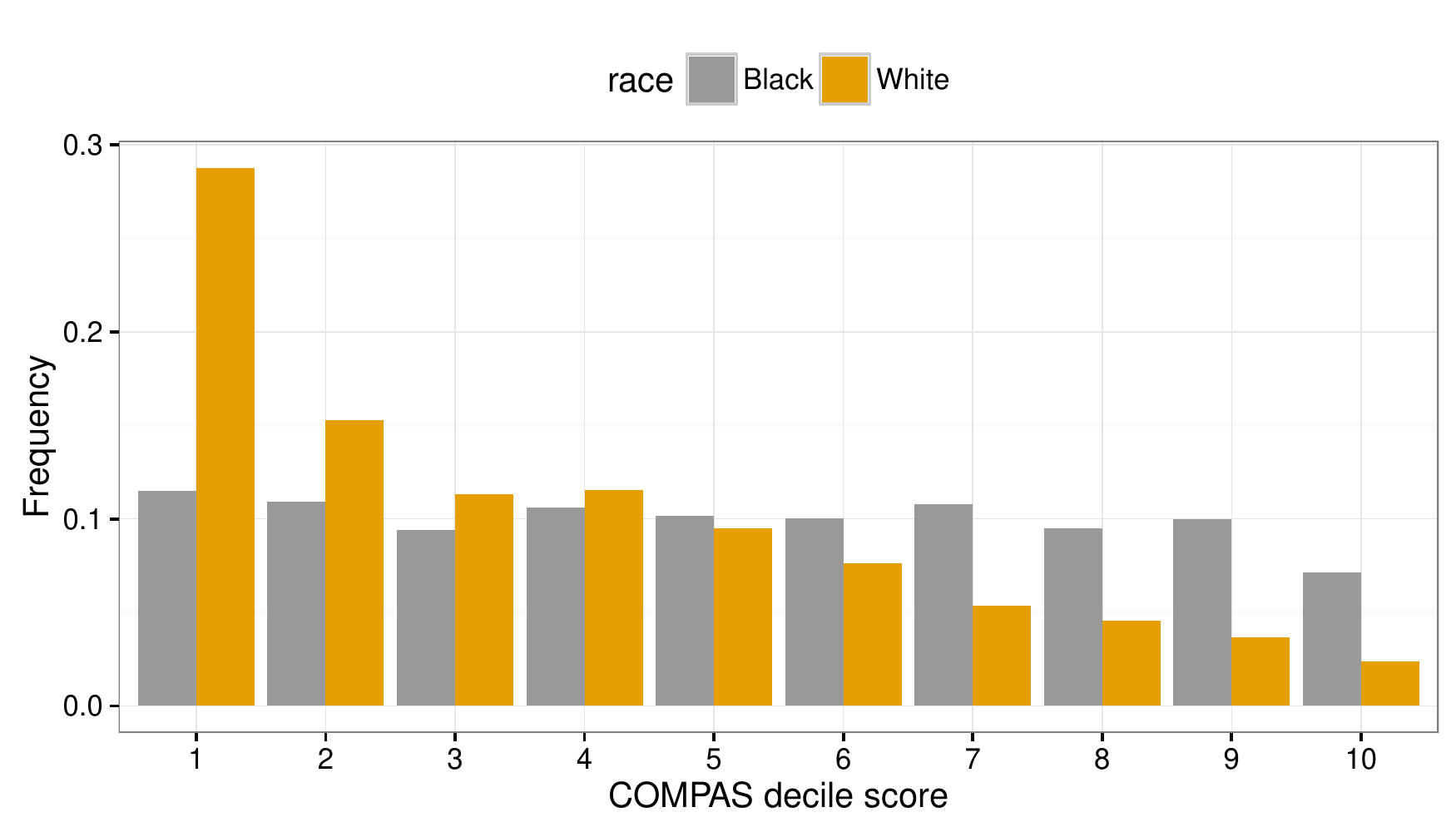}
    \caption{\small COMPAS decile score histograms for Black and White defendants.  Cohen's $d = 0.60$, non-overlap $d_{\mathrm{TV}}(f_b, f_w) = 24.5\%$.}
    \label{fig:score_plot}
 \end{figure}

  In their analysis of the PCRA instrument, \citet{skeem2015risk} remark that some applications of the risk score could create disparate impact due to differences in the  score distributions between black and white offenders.  To summarize the distributional difference in scores between the two groups, the authors report a Cohen's $d$ of $0.34$, with a corresponding non-overlap of 13.5\%.  A natural question to ask is whether the level of disparity in sentence duration, $\Delta$, is in some sense closely related to such measures of distributional difference.  With a small generalization of the \emph{\% non-overlap} measure, we can answer this question in the affirmative.  
  
  The \% non-overlap of two distributions is generally calculated assuming both distributions are normal, and thus has a one-to-one correspondence to Cohen's $d$ \cite{cohen1988}.\footnote{$d = \frac{\bar S_b - \bar S_w}{SD}$, where $SD$ is a pooled estimate of standard deviation.} However, as we can see from Figure~\ref{fig:score_plot}, the COMPAS decile score is far from being normally distributed in either group.  A more reasonable way to calculate \% non-overlap in such cases is to note that in the Gaussian case \% non-overlap is equivalent to the total variation distance.  Letting $f_{r,y}(s)$ denote the score distribution among individuals in group $r$ with recidivism outcome $y$, one can establish the following sharp bound on $\Delta$.
\begin{prop}[Percent overlap bound]
  Under the MinMax policy,
  \[
  \Delta(y_1, y_2) \le  (t_\mathrm{max} - t_\mathrm{min}) d_\mathrm{TV}(f_{b,y_1}, f_{w,y_2}).
  \]
  \label{prop:overlap}
\end{prop}

\vspace{-3em}

This result is simple to understand.  When there is some non-overlap between the score distributions for two groups, the worst case scenario is that the non-overlap is entirely due to mass shifting from scores below $s_\mathrm{HR}$ to those above $s_\mathrm{HR}$.  In such cases, the inequality becomes an equality.

\subsection{Empirical results} \label{sec:empirical}

In this section we present some empirical results based on two hypothetical sentencing rules: the \emph{MinMax} rule introduced in the previous section, and the \emph{Interpolation} rule, which we will introduce below.  Though the offenders in our data set come from Broward County, Florida, our empirical analysis is modelled on the sentencing guidelines of the State of Pennsylvania. 

The penalty ranges $t_\mathrm{min}$ and $t_\mathrm{max}$ are selected by approximately matching each offender's charge degree (M2 - F1) to a sentence range in Pennsylvania's Basic Sentencing Matrix (PA Code \textsection 303.16).  This matrix provides sentence ranges based on the charge degree for the current offense and the defendant's prior record score (0 - 5+).  We do not have enough information in the Broward County data to reliably assign a prior record score for each individual.  Our results are based on using the sentencing range corresponding to a prior record score of $1$ for all defendants in the data.  

\begin{figure}[ht]
 \centering
 \begin{subfigure}[t]{0.48\linewidth}
   \includegraphics[width=\textwidth]{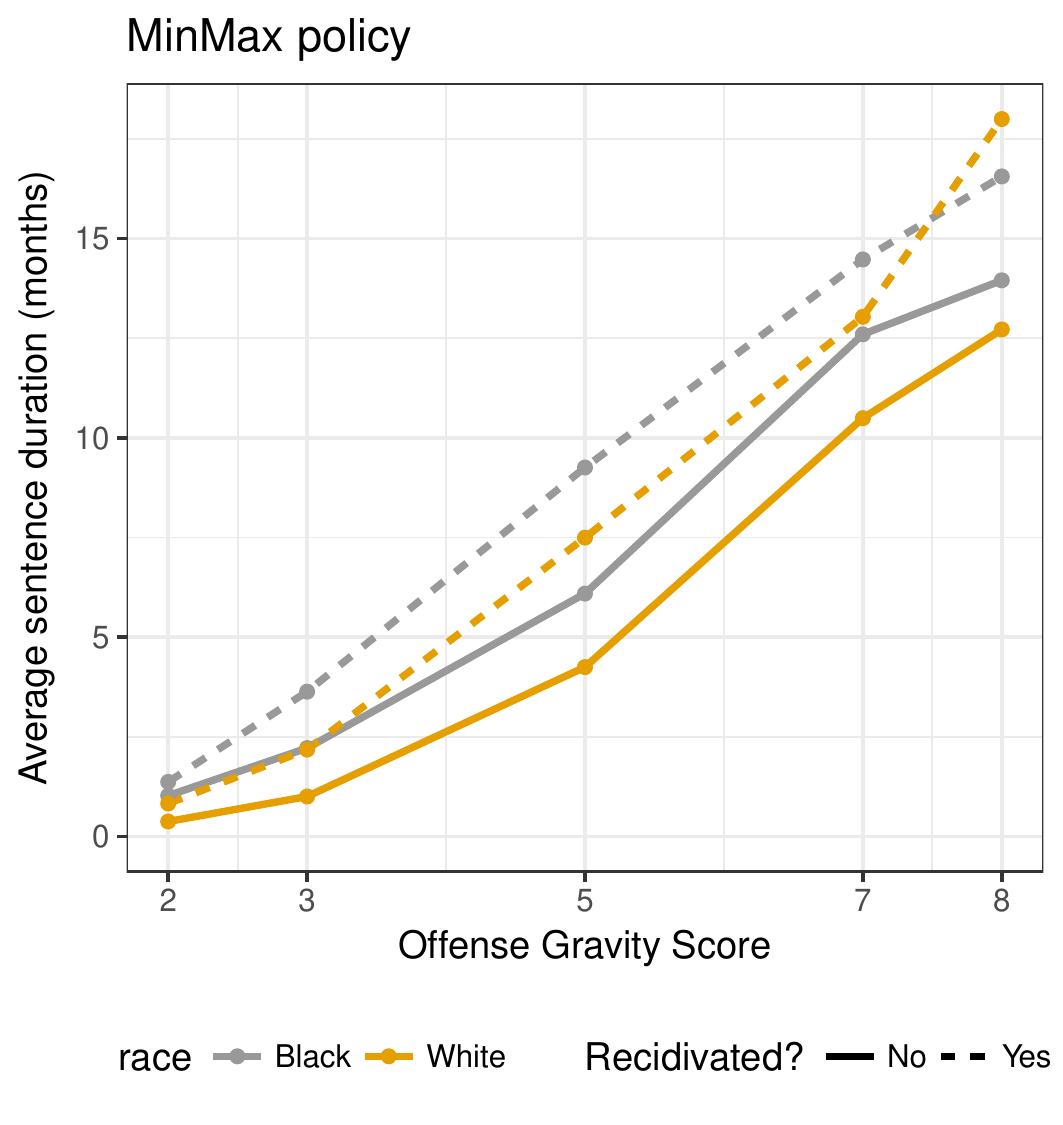}
 \end{subfigure}
 \hspace{1em}
 \begin{subfigure}[t]{0.48\linewidth}
   \includegraphics[width=\textwidth]{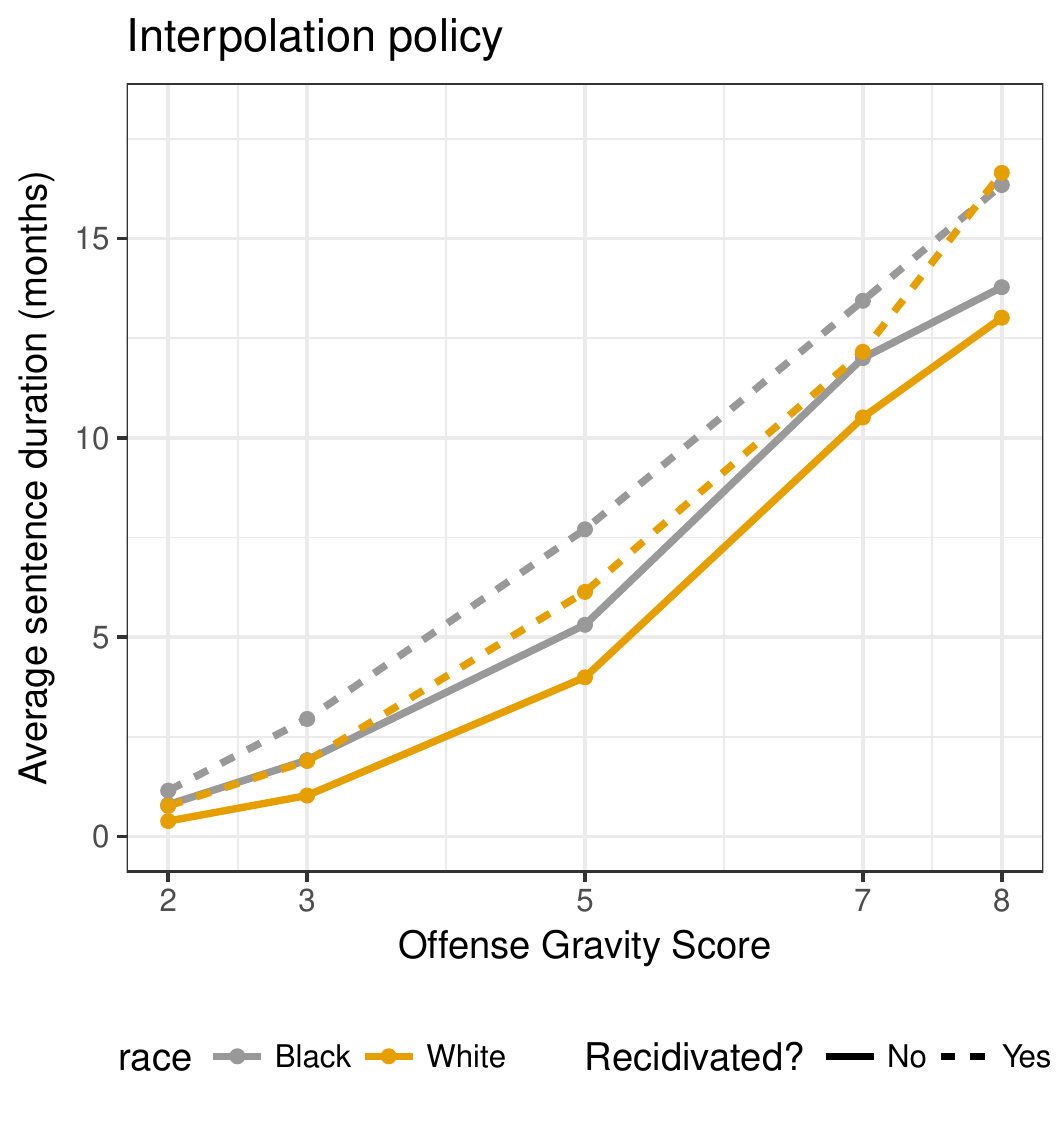}
 \end{subfigure}

 \vspace{1em}
 \caption{Average sentences under the hypothetical sentencing policies described in Section~\ref{sec:empirical}.  The mapping between the $x$-axis variable and the offender's charge degree is given in Table~\ref{tab:ogsmap}.  For all OGS levels except 8, observed differences in average sentence are statistically significant at the $0.01$ level.}
 \label{fig:empirical}
\end{figure}

Figure~\ref{fig:empirical} shows the expected sentences for black and white defendants broken down by observed recidivism outcome.  The $x$-axis in these figures is taken to be the offense gravity score, which for the purpose of this analysis is mapped to charge degree as indicated in Table~\ref{tab:ogsmap}.

\begin{table}[h]
  \centering
    \begin{tabular}{|r|ccccc|}
      \hline
      Offense gravity score & 2 & 3 & 5 & 7 & 8 \\ 
      Charge Degree & (M2) & (M1) & (F3) & (F2) & (F1) \\ 
      \hline
    \end{tabular}
\caption{Mapping between offense gravity score and charge degree used in the empirical analysis.}
\label{tab:ogsmap}
\end{table}

Results are shown for both the MinMax policy introduced earlier in this section, and the Interpolation policy, which is given by
\begin{equation}
    T_{\mathrm{Int}}(s) = t_{\mathrm{min}} + \frac{s - 1}{9}(t_{\mathrm{max}} - t_{\mathrm{min}}).
\end{equation}
Unlike the MinMax policy, which is based on the coarsened score, the Interpolation policy assigns sentences by linearly interpolating between $t_\mathrm{min}$ and $t_\mathrm{max}$ based on the assigned decile score.  We see that under both policies there are consistent trends in the expected sentences.  Black defendants are observed to receive higher sentences than white defendants both within the non-recidivating subgroup and the recidivating subgroup (except in the F1 charge degree category, where sample sizes are small and results are non-significant).  Since white defendants have higher false negative rates and lower false positive rates than black defendants, the empirical results are consistent with the theoretical results presented earlier in this section.  

\section{Revisiting predictive parity}

In this final section we revisit the notion of predictive parity and further discuss its implications for general classifiers.  We know from equation~\eqref{eq:fpr_fnr} that when the positive predictive values are constrained to be equal but the prevalences differ across groups, the false positive and false negative rates cannot both be equal across those groups.  While we have no direct control over recidivism prevalence, we do have some control over the PPV and error rates of our classifiers.  At least in principle, we are free to tune our classifiers in any of the following ways:

\begin{enumerate}[(i)]
  \item Allow unequal false negative rates to retain equal PPV's and achieve equal false positive rates
  \item Allow unequal false positive rates to retain equal PPV's and achieve equal false negative rates
  \item Allow unequal PPV's to achieve equal false positive and false negative rates
\end{enumerate}

\begin{figure}[ht]
  \centering
  \includegraphics[width=0.85\linewidth]{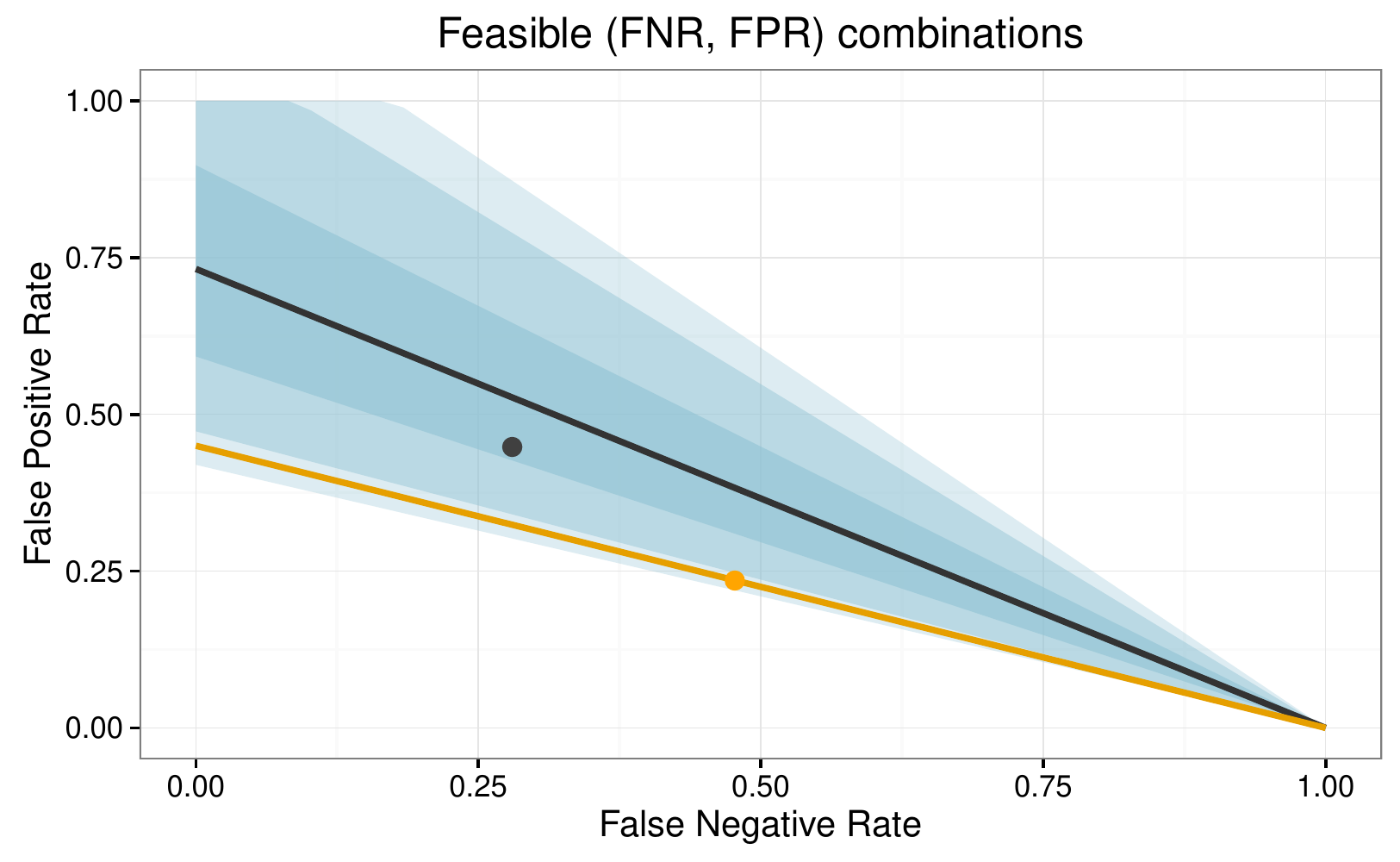}
  \caption{The two points represent the observed values of (FNR, FPR) for Black and White defendants.  The orange line represents feasible values of (FNR, FPR) for White defendants when the prevalence $p_w$ and $\mathrm{PPV}_w$ are both held fixed at their observed values in Table~\ref{tab:confusion}.  The dark grey line represents feasible values of $(\fnr_b, \fpr_b)$ when the prevalence $p_b$ is held fixed at the observed value and $\mathrm{PPV}_b$ is set equal to the observed value of $\mathrm{PPV}_w = 0.591$.  \emph{Nested} shaded regions correspond to feasible values of $(\fnr_b, \fpr_b)$ if we allow $\mathrm{PPV}_b$ to vary under the constraint $|\mathrm{PPV}_b - 0.591| < \delta$, with $\delta \in \{0.05, 0.1, 0.125\}$.  The smaller $\delta$, the smaller the feasible region. }
  \label{fig:shadeplot}
\end{figure}

Figure~\ref{fig:shadeplot} helps to put these trade-offs into perspective.  From~\eqref{eq:fpr_fnr}, we can see that FPR is a linear function of FNR under constraints on PPV and $p$.  This means that, if PPV is fixed at a given value, tuning strategy (i) may require a very large increase in $\fnr$ in order to balance FPR.  The black line shows feasible combinations of $(\fnr_b, \fpr_b)$ when $\mathrm{PPV}_b$ is forced to equal the observed value $\mathrm{PPV}_w = 0.591$.  We can see that to get $\fpr_b$ to match  $\fpr_w$, we would need to increase $\fnr_b$ to around $0.7$, which would be a substantial drop in accuracy.  In view of Corollaries \ref{cor:nonrecid} and \ref{cor:recid} Strategies (i) and (ii) may generally be undesirable because while they reduce disparate impact for one subgroup (e.g., among non-recidivists), they may increase it in the other.  

The preferred approach, at least in some cases, may be to pursue strategy (iii).  This amounts to using a score $S$ that does not satisfy predictive parity in the first place, but can also be achieved by allowing the high-risk cutoff $s_{\mathrm{HR}, r}$ to differ across groups.  The shaded regions in Figure~\ref{fig:shadeplot} show feasible values of $(\fnr_b, \fpr_b)$ when we allow $\mathrm{PPV}_b$ to be within some  $\delta$ of the observed value of $\mathrm{PPV}_w$.  We can see that even at small values of $\delta$ the feasible region is quite large.  

\section{Discussion}

The primary contribution of this paper was to show how disparate impact can result from the use of a recidivism prediction instrument that is known to satisfy the fairness criterion of predictive parity.  Our analysis focussed on the simple setting where a binary risk assessment was used to inform a binary penalty policy.  While all of the formulas have natural analogs in the non-binary score and penalty setting, we find that many of the salient features are already present in the analysis of the simpler binary-binary problem.  

A key limitation of our analysis stems from potential biases in the observed data that may affect our ability to draw valid inferences concerning the fairness of an RPI.  Throughout this paper we have implicitly operated under the assumption that the observed recidivism outcome $Y$ is a suitable outcome measure for the purpose of assessing the fairness properties of a recidivism prediction instrument.  However, the true outcome of interest in this context is \emph{reoffense}, which is not what we observe.  In the latest statistics released by the Federal Bureau of Investigation\cite{fbi-clearance-2016}, it is reported that 46\% of violent crimes and 19.4\% of property crimes were successfully cleared by law enforcement agencies.  Many criminal offenders are simply never identified.  It is therefore possible that a non-negligible fraction of the individuals in our data for whom we observed $Y = 0$ did in truth reoffend.  If this is indeed the case, and if there are group differences in the rates at which offenders are caught, the findings of empirical fairness assessments may be misleading.  Understanding how such forms of data bias affect the ability to assess instruments with respect to different fairness criteria is a subject of our ongoing research efforts.

\section{Conclusion}

In closing, we would like to note that there is a large body of literature showing that data-driven risk assessment instruments tend to be more accurate than professional human judgements \citep{meehl1954clinical, grove2000clinical}, and investigating whether human-driven decisions are themselves prone to exhibiting racial bias \citep{anwar2012testing, sweeney1992influence}.  We should not abandon the data-driven approach on the basis of negative headlines.  Rather, we need to work to ensure that the instruments we use are demonstrably free from the kinds of biases that could lead to disparate impact in the specific contexts in which they are to be applied.

\appendix 

\section{Proofs} \label{appendix:proofs}

\begin{proof}[Proof of Proposition~\ref{prop:main}]  To simplify notation, we let $HR$ denote the event $\{S > s_\mathrm{HR}\}$.
  \begin{align*}
  \E(\Delta(y_1, y_2)) &= \E(T \mid R = b, Y = y_1) - \E(T \mid R = w, Y = y_2) \\
    &= t_\mathrm{max}\P(HR \mid R = b, Y = y_1) + t_\mathrm{min}(1 - \P(HR \mid R = b, Y = y_1)) \\
    & \qquad - t_\mathrm{max}\P(HR \mid R = w, Y = y_2) - t_\mathrm{min}(1 - \P(HR \mid R = w, Y = y_2)) \\
    &= t_\mathrm{max}(\P(HR \mid R = b, Y = y_1) - \P(HR \mid R = w, Y = y_2)) \\      & \qquad + t_\mathrm{min}(\P(HR \mid R = w, Y = y_2) - \P(HR \mid R = b, Y = y_1)) \\
    &= (t_\mathrm{max} - t_\mathrm{min})(\P(HR \mid R = b, Y = y_1) - \P(HR \mid R = w, Y = y_2))
  \end{align*}
\end{proof}

\begin{proof}[Proof of Proposition~\ref{prop:overlap}]
  By definition of total variation distance, for any event $A$,
  \[
  \left|(\P(A\mid R = b, Y = y_1) - \P(A \mid R = w, Y = y_2))\right| \le 
    d_\mathrm{TV}(f_{b,y_1}, f_{w,y_2})
  \]
  Applying this inequality to Proposition~\ref{prop:main} with $A = \{S_c = \mathrm{HR}\}$ gives
  \begin{align*}
    \E(\Delta(y_1, y_2)) &= (t_\mathrm{max} - t_\mathrm{min})(\P(HR \mid R = b, Y = y_1) - \P(HR \mid R = w, Y = y_2)) \\
    &\le (t_\mathrm{max} - t_\mathrm{min}) d_\mathrm{TV}(f_{b,y_1}, f_{w,y_2})
  \end{align*}
\end{proof}

\section{Covariate-adjusted false positive rates} \label{appendix:logistic}

In this section we present the results of a logistic regression analysis that we conducted in order to assess whether the observed differences in false positive rates between black and white defendants can be entirely accounted for by other covariates.  We find that adjusting for covariates decreases the gap, but it nevertheless remains large and statistically significant.

For the purpose of this analysis we consider only the subset of defendants who \emph{do not} recidivate.  The outcome variable for the logistic regression is taken to be 
\[
y = \begin{cases}
1, & S > 4 \\
0, & S \le 4
\end{cases},
\]
where $S$ denotes the COMPAS decile score.  In this setup, $y = 0$ denotes a True Negative and $y = 1$ denotes a False Positive.  Statistically significant positive coefficient estimates correspond to variables associated with increased likelihood of false positives.

Table~\ref{tab:race_alone} shows the results of regressing $y$ on race alone.  The coefficient of race in this model is large, positive, and statistically significant.  Without adjusting for other covariates, the odds that a non-recidivating Black defendant receives a high-risk assessment are $e^{0.976}= 2.6$ times higher than those of a White defendant.

Table~\ref{tab:race_adjusted} shows the results of regressing $y$ on race, age, gender, number of priors, and charge degree.  The coefficient of race is smaller than it was in the un-adjusted model, but it is nevertheless large and statistically significant.  Even after adjusting for these other factors, the odds that a non-recidivating Black defendant receives a high-risk assessment are $e^{0.547} = 1.72$ times higher than those of a White defendant.

  \begin{table}[ht]
  \centering
    \begin{tabular}{rrrrr}
      \hline
     & Estimate & Std. Error & z value & Pr($>$$|$z$|$) \\ 
      \hline
    (Intercept) & -1.183 & 0.061 & -19.33 & 0.0000 \\ 
      raceBlack & 0.976 & 0.077 & 12.60 & 0.0000 \\ 
       \hline
    \end{tabular}
    \caption{Logistic regression with race alone.}
    \label{tab:race_alone}
  \end{table}

  \begin{table}[ht]
    \centering
  \begin{tabular}{rrrrr}
    \hline
   & Estimate & Std. Error & z value & Pr($>$$|$z$|$) \\ 
    \hline
  (Intercept) & 1.397 & 0.176 & 7.92 & 0.0000 \\ 
    raceBlack & \textbf{0.547} & 0.087 & 6.30 & 0.0000 \\ 
    Age & -0.079 & 0.005 & -17.48 & 0.0000 \\ 
    sexMale & -0.291 & 0.098 & -2.97 & 0.0030 \\ 
    Number of Priors & 0.283 & 0.016 & 17.78 & 0.0000 \\ 
    chargeMisdemeanor & -0.109 & 0.088 & -1.25 & 0.2123 \\  
     \hline
  \end{tabular}
  \caption{Logistic regression with race and other covariates that may be associated with recidivism}
  \label{tab:race_adjusted}
  \end{table}
  
\newpage 
\bibliographystyle{unsrtnat}
\bibliography{recidivism}
\end{document}